%% file: main_imag.tex
\documentclass[aps,10pt,showpacs,prl,amsmath,amssymb,floatfix,superscriptaddress,longbibliography,twocolumn]{revtex4-2}

\input{preamble}

\begin{document}

\title{Unitary-invariant witnesses of quantum imaginarity}

\author{Carlos Fernandes}
\email{carlos.fernandes@inl.int}
\affiliation{International Iberian Nanotechnology Laboratory (INL), Av. Mestre Jos\'{e} Veiga, 4715-330 Braga, Portugal}
\affiliation{Centro de F\'{i}sica, Universidade do Minho, Braga 4710-057, Portugal}

\author{Rafael Wagner}
\email{rafael.wagner@inl.int}
\affiliation{International Iberian Nanotechnology Laboratory (INL), Av. Mestre Jos\'{e} Veiga, 4715-330 Braga, Portugal}
\affiliation{Centro de F\'{i}sica, Universidade do Minho, Braga 4710-057, Portugal}

\author{Leonardo Novo}
\email{leonardo.novo@inl.int}
\affiliation{International Iberian Nanotechnology Laboratory (INL), Av. Mestre Jos\'{e} Veiga, 4715-330 Braga, Portugal}

\author{Ernesto F. Galv\~ao}
\email{ernesto.galvao@inl.int}
\affiliation{International Iberian Nanotechnology Laboratory (INL), Av. Mestre Jos\'{e} Veiga, 4715-330 Braga, Portugal}
\affiliation{Instituto de F\'{i}sica, Universidade Federal Fluminense, Niter\'{o}i -- RJ, Brazil}

\date{\today}

\begin{abstract}
   Quantum theory is traditionally formulated using complex numbers. This imaginarity of quantum theory has been quantified as a resource with applications in discrimination tasks, pseudorandomness generation, and quantum metrology. Here we propose witnesses for imaginarity that are basis-independent, relying on measurements of unitary-invariant properties of sets of states. For 3 pure states, we completely characterize the invariant values attainable by quantum theory, and give a partial characterization for 4 pure states. We show that simple pairwise overlap measurements suffice to witness imaginarity of sets of 4 states, but not for sets of 3. Our witnesses are experimentally friendly, opening up a new path for measuring and using imaginarity as a resource.
\end{abstract}

\maketitle

\textit{Introduction.}-- Hickey and Gour~\cite{hickey2018quantifying} introduced a quantum resource theory~\cite{chitambar2019quantum} based on the fact that, for a given fixed basis of reference $\mathbb{A} := \{\vert a_i\rangle\}_i\subseteq \mathcal{H}$ some states have coherences~\cite{baumgratz2014quantifying} necessarily described by complex numbers $\langle a_i |\rho |a_j\rangle \notin \mathbb{R}$. They termed this the \emph{imaginarity} of a quantum state. In spite of this rigorous resource-theoretic treatment, its usefulness was initially unclear. Some early results in the literature suggested that imaginarity had no relevant role, as, for instance, various real-only formulations of quantum mechanics were known~\cite{stueckelberg1960quantum,stueckelberg1961quantum,antoniya2013real,hardy2011limited,wootters2010entanglement,wootters2015optimal}. Moreover, quantum theory restricted to real amplitudes was shown to be universal for quantum computation~\cite{rudolph2002rebit,aharonov2003simple}, and  capable of reproducing the statistics of any Bell experiment~\cite{McKague2009simulating}. Any generic complex-valued quantum computation can be evaluated using only real amplitudes by using an extra auxiliary qubit, enlarging the system from $\mathcal{H} \mapsto \mathbb{C}^2 \otimes \mathcal{H}$. In such a way, free states $\rho^{\mathbb{R}}$ and observables $H^{\mathbb{R}}$ defined in the larger system reproduce all predictions of the resourceful states $\rho$ and observables $H$ via $\text{Tr}(\rho H) = \text{Tr}(\rho^{\mathbb{R}} H^{\mathbb{R}})$. 

Surprisingly, Renou \textit{et al.}~\cite{renou2021quantum} found an experimental scenario that cannot be exactly modelled using quantum theory with real-valued amplitudes only. These findings have been experimentally verified~\cite{li2022testing,chen2022ruling,wu2022experimental}, and a few other similar scenarios have since been proposed~\cite{bednorz2022optimal,yao2024proposals}. The requirements to be met for such experiments are demanding, since they are based on network nonlocality scenarios. For example, every source needs to be independent, and measurements 
 from different parties must be space-like separated. In simple terms, requiring this independence -- and therefore a certain tensor product structure to be preserved -- guarantees that the trick of extending $\mathcal{H} \mapsto \mathbb{C}^2 \otimes \mathcal{H}$ fails. 

The foundational results of Renou \emph{et. al.} led to a burst of investigations on the imaginarity of quantum theory. It has been shown that imaginarity can be made operationally meaningful in a precise sense~\cite{wu2021operational,wu2021resource,wu2023resource}. Imaginarity of quantum theory was shown to have crucial effects in certain discrimination tasks~\cite{wu2021operational,herzog2002minimum}, hiding and masking~\cite{zhu2021hidingmasking}, machine learning~\cite{sajjan2023imaginary}, pseudorandomness~\cite{haug2023pseudorandom}, multiparameter metrology~\cite{carollo2018uhlmann,carollo2019quantumness,miyazaki2022imaginarityfree}, outcome statistics of linear-optical experiments~\cite{jones2022distinguishability,menssen2017distinguishability,shchesnovich2018collective}, Kirkwood-Dirac quasiprobability distributions~\cite{wagner2023quantum,budiyono2023operational,budiyono2023quantifying,budiyono2023quantum}, and weak-value theory~\cite{wagner2023anomalous,kedem2012usingtechnical,dixon2009ultrasensitive,hosten2008observation,brunner2010measuringsmall,hofmann2011uncertainty,kunjwal2019anomalous}.

Network nonlocality witnesses of the type used to demonstrate the need for imaginarity are found via relaxations of semidefinite programing problems~\cite{tavakoli2023semidefinite}. They are basis-independent, dimension-independent, device-independent, and independent of extra assumptions such as the possibility of local tomography. The requirement of only weak assumptions certainly makes the foundational argument more compelling, but it has a price. The possible witnesses and experimental scenarios  are few and challenging, requiring precise clock control of independent sources, space-like separation, and measurements performed by many parties. 

An alternative way of witnessing imaginarity was pointed out in Refs.~\cite{oszmaniec2021measuring, wagner2023quantum}, and involves measurements of unitary-invariant quantities known in the literature as Bargmann invariants, or multivariate traces~\cite{bargmann1964note,oszmaniec2021measuring,quek2023multivariate}. Given a tuple of states $ \varrho = (\rho_1,\dots,\rho_n) \in \mathcal{D}(\mathcal{H})^n$, its Bargmann invariant is the (in general) complex-valued quantity $\text{Tr}(\rho_1\rho_2 \dots \rho_n)$, measurable by circuit families proposed in Refs.~\cite{oszmaniec2021measuring, quek2023multivariate}. {A function $f:\mathcal{D}(\mathcal{H})^n \to \mathbb{C}$ of the tuple of states  $\varrho$ is said to be unitary-invariant if  $f(U \varrho U^{\dagger}) = f(\varrho)$ and for any unitary $U$~\footnote{Here, $U \varrho U^\dagger := (U \rho_1 U^\dagger, \dots, U \rho_n U^\dagger)$, for any $\varrho = (\rho_1,\dots,\rho_n)$.}. Bargmann invariants clearly have that property.} When the estimated invariants have a non-zero imaginary part, at least one state in $\varrho$ must have quantum imaginarity. 

In this work we develop the theory characterizing the role of these invariants as witnesses of quantum imaginarity. They are experimentally promising, and their usage shares many of the benefits of using the Bell nonlocal network approach of Renou \emph{et. al.}~\cite{renou2021quantum} -- independence of basis and dimension, with no need for further assumptions such as local tomography. Our tests are \emph{not} device-independent, and assume the statistics to correspond to values of unitary-invariants, measurable in relatively straightforward experiments which we briefly describe.

We start with the foundational question of what complex values 3-state invariants can assume, and solve this problem completely by characterizing the region in the complex plane where all possible third-order invariants must lie. For this scenario with 3 pure states, we prove that pairwise overlaps cannot witness imaginarity. This is because any pairwise overlap triple, arising out of arbitrary triples of quantum states, can be alternatively obtained from states whose amplitudes are real-valued.

Besides the complete characterization of which values 3-state invariants can take, we obtain a partial characterization of the 4-state scenario. We show that, unlike the 3-state case, the imaginarity of 4 states \emph{can} be witnessed with pairwise overlap measurements only. This allows for simplified experimental tests to witness imaginarity of four states, as pairwise overlaps are simpler to measure than  higher-order Bargmann invariants. We also briefly discuss implications of our results for the quantitative description of weak values, Kirkwood-Dirac quasi-probability distributions, multiphoton indistinguishability, and geometric phases. 

\textit{Background: Imaginarity of quantum states.}--If we fix a quantum system $\mathcal{H}$ and a reference basis $\mathbb{A} = \{\vert i \rangle \}_i$, we call states $\rho^{\mathbb{R}}$ \emph{real-representable}, or simply \emph{free states}, if $\forall i,j,\langle i |\rho^{\mathbb{R}}|j\rangle \in \mathbb{R}$. We write this set as $\mathcal{R}(\mathcal{H},\mathbb{A})$. Each choice $(\mathcal{H},\mathbb{A})$ defines a section of quantum theory that can be treated as free within the resource theory of imaginarity~\cite{hickey2018quantifying,wu2021operational}. It is important to note that, in this treatment of free resources, fixing both $\mathcal{H}$ and $\mathbb{A}$ is crucial to make the theory non-trivial, as every single state $\rho \in \mathcal{R}(\mathcal{H}, \mathbb{A}_\rho)$ is real with respect to its spectral basis $\mathbb{A}_\rho$. Moreover, as mentioned in the introduction, suppose that we fix, for each dimension $d$, a certain canonical basis $\mathbb{A}$ that characterizes the free resources. If $\rho \in \mathcal{D}(\mathcal{H})$ is not free, then the simple operation
\begin{equation*}
    \rho \mapsto \rho^{\mathbb{R}} = \frac{1}{2}\left(\begin{matrix}
        \text{Re}[\rho] & - \text{Im}[\rho]\\
        \text{Im}[\rho] & \text{Re}[\rho]
    \end{matrix}\right) \in \mathcal{R}(\mathbb{C}^2 \otimes \mathcal{H},\mathbb{A})
\end{equation*}
sends any non-free state into a real-only one, in a higher dimensional Hilbert space. This encoding into real-amplitude states enables universal quantum computation~\cite{rudolph2002rebit,aharonov2003simple}, and arbitrary Bell inequality violations~\cite{McKague2009simulating}. 

We now relax the assumption of fixing a basis. This was first done in Ref.~\cite{miyazaki2022imaginarityfree} with the techniques developed in Ref.~\cite{designolle2021set} for coherence theory. If instead of considering a single state $\rho$ one considers a  \emph{set} of states $\{\rho_i\}_{i=1}^n \subset \mathcal{D}(\mathcal{H})$, one does not need to fix a basis of reference. Our free objects are now sets of states all of which have real amplitudes in some basis. Hence a set $\{\rho_i\}_i$ is free if $\exists U_{\mathbb{A}}:\mathcal{H} \to \mathcal{H}$ unitary such that $\forall i, U_{\mathbb{A}}\rho_i U_{\mathbb{A}}^\dagger = \rho_i^{\mathbb{R}}$, where $\rho_i^{\mathbb{R}}$ is real-represented with respect to some basis $\mathbb{A}$. In this manner, whether a set of states can be real-represented is a manifestly unitary-invariant property of the set. \emph{Set imaginarity}~\cite{miyazaki2022imaginarityfree} is the term used to describe a state set which cannot be taken unitarily to a set having only real amplitudes.

\textit{Quantum realizability of unitary-invariants.}-- We start stating our results by noting the simple fact that $\text{Im}[\Delta] \neq 0$ implies set imaginarity of \emph{every} tuple $\varrho \in \mathcal{D}(\mathcal{H})^n$ where  $\Delta = \Delta_n(\varrho) = \text{Tr}(\rho_1 \rho_2 \dots \rho_{n-1}\rho_n)$. Deciding if, for certain values $\Delta \in \mathbb{C}$, there exist tuples of quantum states $\varrho$ such that $\Delta_n(\varrho) = \Delta$  constitutes a quantum realizability problem~\cite{fraser2023realization}. 

To fix the notation, we will consider tuples of pure states $\Psi$. It is well-known that two tuples $\Psi, \Phi \in \mathcal{P}(\mathcal{H})^n$, with $\mathcal{P}(\mathcal{H}) = \{\vert \psi\rangle \langle \psi \vert : \vert \psi \rangle \in \mathcal{H}\}$, are unitarily equivalent if, and only if, their associated Gram matrices $G_\Psi = G_\Phi$ are equal \cite{chien2016projective,halperin1962onthe}. The Gram matrix $G_V$ of a tuple of vectors $V$ is given by $(G_V)_{ij} = \langle v_i|v_j\rangle$. However, inner-products are not physically meaningful, as they are not invariant under arbitrary choices of {global phases $e^{i\theta_k}\vert v_k\rangle$ for each vector in $V$, which is known as a \emph{choice of gauge} for the tuple}.

There are, nevertheless,  representations of $G_{\Psi}$ where every matrix element is written only in terms of Bargmann invariants{, which are invariant by choices of gauge}. Refs.~\cite{chien2016projective,oszmaniec2021measuring} introduced a procedure to {find such representations. It} is possible to write $G_\Psi=G(\pmb{\Delta}_\Psi)$, { in terms of a minimal \emph{tuple} of parameters $\pmb{\Delta}_\Psi$, which are functions of Bargmann invariants and thus independent of the gauge choice}. The general theory for obtaining the matrices  $G(\pmb{\Delta}_\Psi)$ can be found in Ref.~\cite{oszmaniec2021measuring}. For our purposes it will suffice to consider only two constructions, one resulting in a matrix of the form of  Eq.~\eqref{eq: Bargmann matrix} for the case of three states, and  another of the form discussed in the Appendix for the case of four states. 

The quantum realizability problem we will consider can then be stated as follows: Given a Hermitian matrix $H$ that we term \emph{candidate}, decide if there exists some tuple of states $\Psi$ for which $ H=G_{\Psi}$.  If so, we say that $ H$ is physically realizable. Because Gram matrices defined from Bargmann invariants are physically meaningful and independent of gauge choice, we will restrict ourselves -- without loss of generality -- to \emph{parameterized candidates} $H(\pmb\Delta)$ that when quantum realizable satisfy $H(\pmb\Delta) = G(\pmb\Delta_\Psi)$. We now discuss this aspect in more detail.

For $\Psi \in \mathcal{P}(\mathcal{H})^3$, for the generic situation where all two-state overlaps are non-zero, a candidate can be chosen to have the form:
\begin{equation}\label{eq: Bargmann matrix}
    H(\pmb{\Delta},\phi) = \left(\begin{matrix}
        1 & \sqrt{\Delta_{12}} & \sqrt{\Delta_{13}}\\
        \sqrt{\Delta_{12}} & 1 & \sqrt{\Delta_{23}}e^{i\phi}\\
        \sqrt{\Delta_{13}} & \sqrt{\Delta_{23}}e^{-i\phi} & 1
    \end{matrix}\right) 
\end{equation}
where $\pmb{\Delta} = (\Delta_{12},\Delta_{13},\Delta_{23}) \in [0,1]^3$ and $\phi \in [0,2\pi)$ are  unconstrained variables. When the candidate above is quantum realizable, there exists some $\Psi$ such that each $\Delta_{ij} = |\langle \psi_i|\psi_j\rangle|^2$ are Bargmann invariants of order two and the phase $\phi$ is the complex phase of the Bargmann invariant of order three $\Delta_{123} = \langle \psi_1|\psi_2\rangle \langle \psi_2|\psi_3\rangle \langle \psi_3\vert \psi_1\rangle$. Note that, in our notation, $G(\pmb{\Delta}_\Psi,\phi_\Psi)$ has the exact same form as Eq.~\eqref{eq: Bargmann matrix}, but having all parameters given by Bargmann invariants. We make a distinction in notation between parameterized candidates $H(\pmb{\Delta})$ for Gram matrices, and $G(\pmb{\Delta}_\Psi)$ that \emph{are} true Gram matrices. 

{We can now ask the quantum realizability question: for what values $\pmb{\Delta}$ and $\phi$ does such a parametrized matrix $H(\pmb{\Delta},\phi)$ give us $G_{\Psi}$ for some $\Psi$?} As such questions are in general, this is subtle. Various values that at first seem acceptable, cannot lead to quantum-realizable matrices. For instance, for $\pmb{\Delta}=(1,0,1)$ and any choice for $\phi$, the candidate $H(\pmb{\Delta},\phi)$ cannot have a quantum realization, since transitivity of equality prevents such overlaps to arise from quantum states~\cite{wagner2022inequalities}. There is, nevertheless, a simple way of characterizing matrices of the form (\ref{eq: Bargmann matrix}) which do correspond to quantum-realizable Gram matrices.

\begin{theorem}[Adapted from Ref.~\cite{chefles2004existence}]\label{theorem: Gram matrix realization}
    Let $H$ be any candidate Hermitian matrix. There exists some $\Psi$ such that $H$ is quantum realizable, i.e., $H = G_\Psi$ iff $H$ positive-semidefinite with principal diagonal $H_{ii} = 1$.
\end{theorem}

Therefore, applying this to the parameterized candidate in Eq.~\eqref{eq: Bargmann matrix}, there exists a quantum realization  $G(\pmb{\Delta}_{\Psi})=H(\pmb\Delta)$ iff there exists a quantum realization for the tuple $(\Delta_{12},\Delta_{13},\Delta_{23},\phi)$. 

We may use this result to gain a better understanding of the possible values for $n$-order Bargmann invariants. Let us define the set of quantum-realizable values for third-order Bargmann invariants as:
\begin{equation}
    \mathcal{B}_3 := \left\{\Delta \in \mathbb{C} : \exists \Psi, \Delta = \langle \psi_1|\psi_2\rangle \langle \psi_2|\psi_3\rangle \langle \psi_3|\psi_1\rangle \right\}.
\end{equation}
It is straightforward to define such a set $\mathcal{B}_n$ also for $n$-th order invariants. Theorem \ref{theorem: Gram matrix realization} enables us to find the border of the set $\mathcal{B}_3$.

\begin{theorem}\label{theorem: gota}
    $\Delta \in \mathcal{B}_3$ if, and only if, $$1-3|\Delta|^{\frac{2}{3}}+2|\Delta| \cos(\phi)\geq 0,$$ where $\Delta = |\Delta|e^{i\phi}$. 
\end{theorem}

\begin{figure}[t]
    \centering
    \includegraphics[width=\columnwidth]{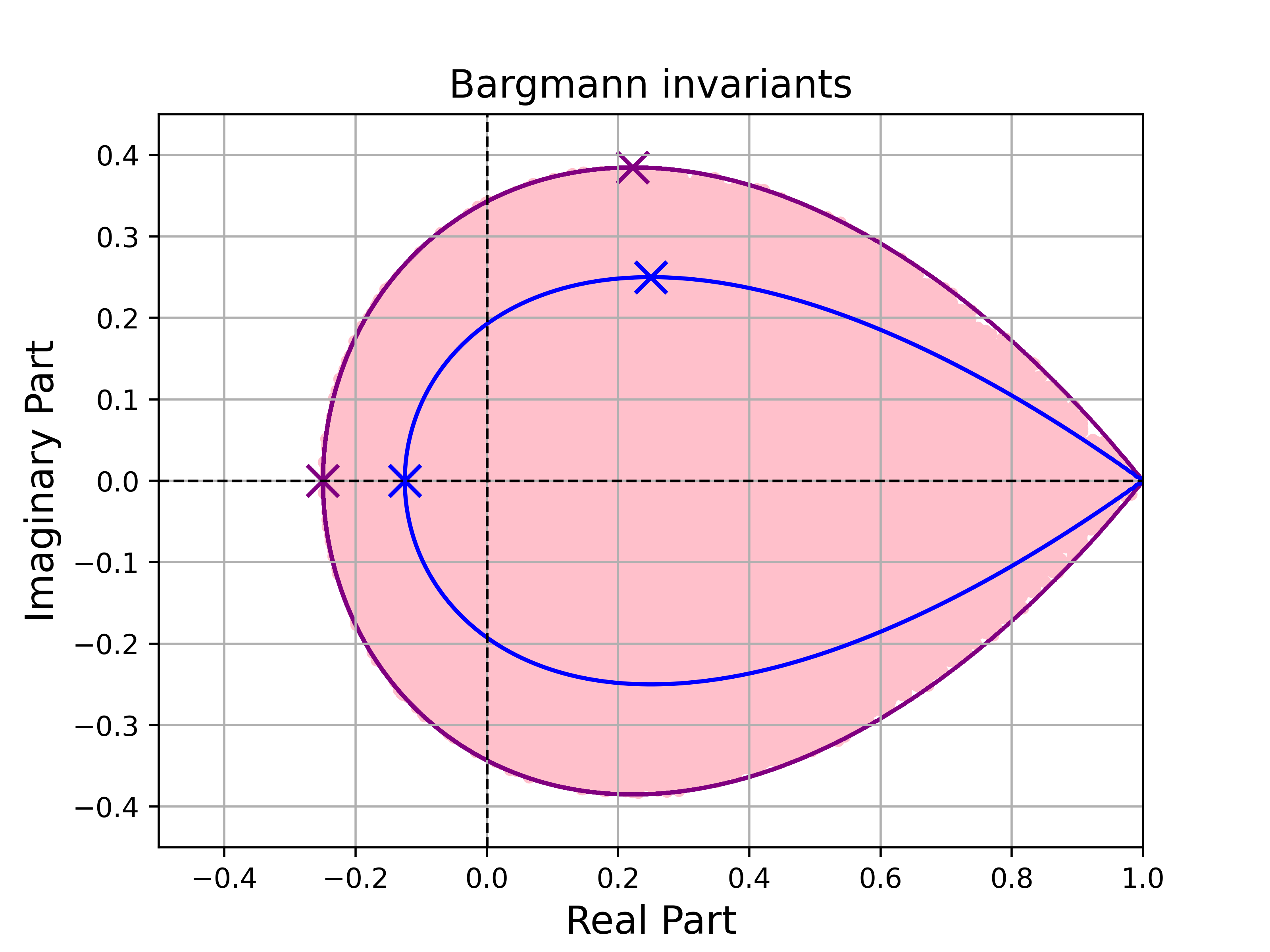}
    \caption{\textbf{Sets of quantum-realizable Bargmann invariants.} (color online) The region within the inner blue boundary is $\mathcal{B}_3$. The outer region describes the set $\mathcal{B}_4\vert_{\mathsf{circ}}$ of Bargmann invariants constructed from tuples $\Psi$ such that $G_{\Psi}$ is a circulant Gram matrix. The points (in pink) are $2.7 \times 10^6$ Bargmann invariants $\Delta_4$ for Haar random tuples $\Psi \in \mathcal{P}(\mathbb{C}^d)^4$ with $d=2,3$, and $4$ ($9 \times 10^5$ for each dimension), which supports the conjecture that $\mathcal{B}_4 = \mathrm{ConvHull}(\mathcal{B}_4 \vert_{\mathsf{circ}})$. The crosses mark invariants with most negative ($\text{Re}[\Delta_3] = -1/8, \text{Re}[\Delta_4] = -1/4$) and largest imaginary parts ($\text{Im}[\Delta_3]= 1/4, \text{Im}[\Delta_4] = 0.38490$).}
    \label{fig: gotas}
\end{figure}

The proof of all new theorems can be found in the Appendix. We plot this region in Fig.~\ref{fig: gotas}. While we do not solve the general problem of analytically characterizing the set $\mathcal{B}_n$ for $n \ge 4$, next we look at (possibly restricted) classes of tuples of states and their invariants. In the Appendix we prove that tuples extremizing the real or imaginary values of $\mathcal{B}_3$ have unitary invariants represented in a Gram matrix $G_{\Psi}$ of \emph{circulant} form~\cite{davis1979circulant}. We define the subsets $\mathcal{B}_n|_{\mathsf{circ}}$ as those $n$-th order Bargmann invariants  $\Delta = \Delta_\Psi$ corresponding to tuples $\Psi \mapsto G_{\Psi}$ where $G_{\Psi}$ is a \emph{circulant} Gram matrix with principal diagonal $(G_\Psi)_{ii}=1$. Using again theorem~\ref{theorem: Gram matrix realization} we can prove a statement similar to theorem~\ref{theorem: gota}.

\begin{theorem}\label{theorem: circulant gota}
    Any $\Delta$ is in the boundary of $\mathcal{B}_4\vert_{\mathsf{circ}}$ if, and only if, 
    \begin{equation*}
        \Delta = \frac{e^{i\phi}}{(\sin(\phi/4) + \cos(\phi/4))^4},
    \end{equation*}
    with $\phi \in [0, 2\pi)$.
\end{theorem}

There are some curious aspects of this result. For example, we can show that the boundary of this set is completely characterized by 4-tuples of \emph{qubit} states, i.e., $\Psi \in \mathcal{P}(\mathbb{C}^2)^4$ (see Appendix). Moreover, numerical evidence strongly suggests that the boundary of the set circulant matrices $\mathcal{B}_4\vert_{\mathsf{circ}}$ is in fact equal to the boundary of the set \emph{all} 4-state invariants $\mathcal{B}_4$ (see Fig.~\ref{fig: gotas}). 

\textit{Physical implications.}-- Theorem~\ref{theorem: gota} provides a complete characterization of the set of all possible third-order invariants $\mathcal{B}_3$ corresponding to pure states in any dimension. Its shape and extrema leads to some implications to various related topics. For example, we start by recalling that a quasiprobability distribution known as the Kirkwood-Dirac (KD) distribution~\cite{lostaglio2022kirkwooddirac,yungerhalpern2018quasi,arvidssonshukur2021conditions} is defined in terms of third-order invariants \cite{wagner2023quantum}. For a given state $\rho$, the distribution value at each discrete phase space point is given by a third-order invariant involving $\rho$ and two states $\vert i\rangle \in I, \vert f \rangle \in F$ from different reference bases $I,F$, given by $\xi(\rho|i,f):=\langle f|i\rangle\langle i|\rho|f\rangle$. Negativity of this distribution has been connected to advantage in quantum metrology~\cite{arvidssonshukur2020quantumadvantage}.  Theorem~\ref{theorem: gota} bounds the optimal imaginary and negative values of this distribution, for \emph{any} choice of state and pair of reference basis. The most negative value of any KD distribution is therefore $-1/8$ and the largest imaginary value is $1/4$ (see Fig.~\ref{fig: gotas}).

On a related note, quantum mechanical weak values \cite{aharonov1988how} are also a function of Bargmann invariants \cite{wagner2023quantum, wagner2023anomalous}. The weak value $\langle P_v \rangle_w$ of a single projection operator $P_v = \vert v \rangle \langle v \vert$ is given by the ratio between a third-order invariant and a second order invariant. Since second-order invariants are real-valued, the negativity/imaginarity of $\langle P_v \rangle_w$, a condition known as weak value anomaly~\cite{wagner2023anomalous}, is directly due to the negativity or imaginarity of third-order invariants. Theorem~\ref{theorem: gota} bounds the contribution to weak value anomaly arising from a third-order invariant with a nonclassical phase $\phi \neq 0$.

Theorem~\ref{theorem: gota} provides strict bounds between the absolute value of a third-order invariant, given as a function of overlaps as $|\Delta_{123}| = \sqrt{\Delta_{12}\Delta_{13}\Delta_{23}}$, and its phase.
In the context of characterization of multiphoton indistinguishability, these phases correspond exactly to so-called collective photonic phases \cite{menssen2017distinguishability,pont2022quantifying,jones2022distinguishability,seron2023boson}. In this setting, each state $\vert \psi\rangle$ in the tuple $\Psi$ describes the internal degrees of freedom of a single photon. In such experiments, one is often interested in maximizing photonic indistinguishability, where all overlaps and higher-order invariants approach the value $1$. In Fig.~\ref{fig: gotas}, this corresponds to a unitary-invariant close to $1$. The analytical boundary of $\mathcal{B}_3$ can then be used to give precise bounds on the possible range of this collective photonic phase, as a function of overlaps, an understanding that is important for the complete characterization of 3-photon indistinguishability. Our preliminary results in Theorem \ref{theorem: circulant gota} start to address the more complex problem of the possible ranges of collective phases for $n>3$ photons.

 As mentioned before, geometric phases acquired after a cyclic sequence of projective measurements are phases of Bargmann invariants, as noticed already in Ref.~\cite{Pancharatnam56}, and elaborated on in later work \cite{Berry09} (for a review see Ref.~\cite{Chruscinski04}). The absolute value squared of the associated Bargmann invariant gives the fraction of events that are postselected in such a sequential experiment. Our theorems \ref{theorem: gota} and \ref{theorem: circulant gota} then give precise trade-offs between this postselection probability and the geometric phases that can be observed.

\textit{Witnessing imaginarity with two-state overlaps.}-- Let us now consider the sets of quantum realizable \emph{overlap tuples}
\begin{equation}
    \mathcal{Q}_{3} := \left\{\pmb{\Delta} \equiv \left(\begin{matrix}\Delta_{12}\\\Delta_{13}\\\Delta_{23}\end{matrix}\right) : \exists \Psi \in \mathcal{P}(\mathcal{H})^3, \pmb\Delta = \pmb\Delta_{\Psi}\right\} .
\end{equation}
Above, $\Delta_{\Psi} = (|\langle \psi_1|\psi_2\rangle|^2, |\langle \psi_1|\psi_3\rangle|^2,|\langle \psi_2|\psi_3\rangle|^2)$. These sets have been widely studied~\cite{galvaobrod2020quantum,wagner2022inequalities,wagner2024coherence,brod2019witnessing} and experimentally investigated~\cite{giordani2020experimental,giordani2021witnesses,giordani2023experimental,pont2022quantifying}. We define also $\mathcal{R}_3\subseteq  \mathcal{Q}_3$ as the set of overlaps realizable by imaginarity-free pure states. Precisely, this is the subset of all $ \pmb\Delta = \pmb\Delta_{\Psi^{\mathbb{R}}} \in \mathcal{Q}_3$ which are realized by a triple of quantum states with real amplitudes, i.e. $\Psi^{\mathbb{R}} \in \mathcal{R}(\mathcal{H},\mathbb{A})^3$ for some Hilbert space $\mathcal{H}$ and basis $\mathbb{A}$. We can now state the following theorem.

\begin{theorem}\label{theorem: 3 overlaps cannot witness}
    Every overlap triple has an imaginarity-free quantum realization, i.e., $\mathcal{Q}_3 = \mathcal{R}_3$.
\end{theorem}

This shows that for the scenario of an arbitrary triple of quantum states, information on the three overlaps $\Delta_{12}, \Delta_{13}, \Delta_{23}$ cannot be used to conclusively witness set imaginarity. Now let $\mathcal{Q}_4, \mathcal{R}_4$ be defined similarly as $\mathcal{Q}_3,\mathcal{R}_3$, but for tuples $\pmb\Delta = (\Delta_{12},\Delta_{13},\Delta_{14},\Delta_{23},\Delta_{24},\Delta_{34})$. We can show that,
\begin{theorem}\label{theorem: overlaps can witness}
    Not every $6$-tuple of overlaps has a quantum realization using 4 states with only real-valued amplitudes. Succinctly, $\mathcal{R}_4 \neq \mathcal{Q}_4$. 
\end{theorem}
To identify tuples of overlaps that are in $\mathcal{Q}_4$ but not in $\mathcal{R}_4$, the strategy is to pick a tuple of overlaps that is quantum-realizable, then consider all possible candidates for Gram matrices that could describe this scenario with given overlaps, but imposing the restriction that all matrix elements must have real-valued $\pm 1$ phases. Even without fixing the gauge, there are only $2^6=64$ such real-valued variant matrices; fixing the gauge reduces this number to $2^3=8$ variants. In the Appendix we report one such example where all real-valued variants have negative eigenvalues, which (by theorem \ref{theorem: Gram matrix realization}), means 4-state imaginarity can be witnessed using overlap values only. 

Such an experimental test of imaginary can be performed by estimating the $6$ overlaps for the $4$ single-qubit states presented in the Appendix, in the proof of Theorem \ref{theorem: overlaps can witness} above. This is a considerably simpler experimental set-up than that required for measuring the imaginarity of three states, using cycle tests \cite{oszmaniec2021measuring}. As we have mentioned before, overlaps can be estimated using either SWAP-tests or Bell-basis projectors in any quantum computational model. Moreover, linear-optical implementations can estimate overlaps from prepare-and-measure statistics~\cite{giordani2023experimental} or from the Hong-Ou-Mandel effect, as in Ref.~\cite{giordani2021witnesses}. These are less demanding than estimating third-order invariants, which requires sources of three single-photons and the interference of these photons in particular three-mode interferometers \cite{menssen2017distinguishability, pont2022quantifying}.

\textit{Discussion and further directions.}-- In this work we have investigated unitary-invariants as witnesses of quantum imaginarity. We characterize the set $\mathcal{B}_3$ of all third-order invariants corresponding to pure states, showing imaginarity cannot be witnessed only by overlap measurements in this case. We also investigate  higher-order invariants, finding the region in the complex plane occupied by invariants of tuples of states whose Gram matrix representations are of circulant form. We propose witnesses of quantum imaginarity of 4 states, and which require only the estimation of their 6 pairwise overlaps. Our results have implications for the quantitative description of geometric phases, Kirkwood-Dirac quasiprobability representations, multiphoton indistinguishability, and quantum-mechanical weak values. They also motivate further investigations on novel experiments to test the imaginarity of quantum theory, and its applications. 

\textit{Acknowledgements.}-- We would like to thank Laurens Walleghem for useful comments and discussions. CF and RW contributed equally to this work. CF acknowledges support from FCT -- Fundação para a Ciência e a Tecnologia (Portugal) through PhD Grant SFRH/BD/150770/2020. RW acknowledges support from FCT -- Fundação para a Ciência e a Tecnologia (Portugal) through PhD Grant SFRH/BD/151199/2021. LN and EFG acknowledge support from FCT-Fundação para a Ciência e a Tecnologia (Portugal) via the Project No. CEECINST/00062/2018. This work was supported by the ERC Advanced Grant QU-BOSS, GA no. 884676.

\bibliography{bibliography}

\appendix
\section{Proof of Theorems~\ref{theorem: gota} and~\ref{theorem: circulant gota}}\label{ap:proofs_th2and3}

\begin{proof}[Proof of theorem~\ref{theorem: gota}]
    Let us consider $H(\pmb\Delta,\phi)$ from Eq.~\eqref{eq: Bargmann matrix} be arbitrary. From Theorem~\ref{theorem: Gram matrix realization} we know that some tuple of states reaching the values of all the invariants defining $H(\pmb\Delta,\phi)$, where $\pmb\Delta = (\Delta_{12},\Delta_{13},\Delta_{23})$,  exists if, and only if, this is a positive-semidefinite matrix. As some of these invariants correspond to third-order invariants, their quantum realizability is determined by the positive-semidefiniteness of $H(\pmb\Delta,\phi)$. Recall that we call $H(\pmb\Delta,\phi)$ quantum realizable iff $H(\pmb\Delta,\phi) = G_\Psi$ for some finite tuple of kets $\Psi$.

    According to Sylvester's criterion, a Hermitian matrix is positive-semidefinite if, and only if, all the determinants of its principal minors are non-negative. This implies in this case the conditions,
    \begin{equation}\label{eq: constraint 1}
        \Delta_{12},\Delta_{13},\Delta_{23} \leq 1
    \end{equation}
    which is always satisfied, and the non-trivial constraint,
    \begin{equation}\label{eq: constraint 2}
        1-\Delta_{12}-\Delta_{13}-\Delta_{23}+2\text{Re}[\Delta_{123}]\geq 0.
    \end{equation}
    Note that $\text{Re}[\Delta_{123}] = \sqrt{\Delta_{12}\Delta_{13}\Delta_{23}}\cos(\phi_{123})$ since $\Delta_{123} = \sqrt{\Delta_{12}\Delta_{13}\Delta_{23}} e^{i\phi_{123}} = |\Delta_{123}|e^{i\phi_{123}}$. Recall that, for any given $\Psi \in \mathcal{P}(\mathcal{H})^3$ quantum realizations of $\Delta_{123}$ are given by 
    \begin{equation*}
        \Delta_{123}(\Psi) = \langle \psi_1|\psi_2 \rangle \langle \psi_2|\psi_3 \rangle \langle \psi_3 | \psi_1\rangle.
    \end{equation*}

    The idea is to maximize the absolute value $|\Delta_{123}|$, but keeping in mind that this value is not independent of the phase $\phi_{123}$. Therefore we maximize $|\Delta_{123}|$ under the constraints given in Eqs.~\eqref{eq: constraint 1} and~\eqref{eq: constraint 2}. There are many such possible choices, for more information we refer the reader to Refs.~\cite{oszmaniec2021measuring,chien2016projective}. Importantly, our results are never dependent of this gauge choice. Each such choice will permute the phase choices inside the matrix $H(\pmb{\Delta},\phi)$. We start noticing that the gradient of $|\Delta_{123}|^2$
    \begin{equation*}
        \frac{\partial |\Delta_{123}|^2}{\partial \Delta_{ij}} = \frac{\Delta_{12}\Delta_{13}\Delta_{23}}{\Delta_{ij}}
    \end{equation*}
    is never zero in the region $0< \Delta_{ij}< 1$, for all $\{i,j\} \in C_3$. Therefore, the maxima must be the points in the boundaries of this region. If any $\Delta_{ij}=0$ then $|\Delta_{123}|^2$ is also zero. If all overlaps are equal to $1$ we have that Eq.~\eqref{eq: constraint 2} becomes the condition $-2+2\cos(\phi_{123})\geq 0$. This is only satisfied for $\phi_{123} = 0$. For finding the maxima that will return non-trivial bounds, we use the method of Lagrange multipliers ($\Delta \equiv \Delta_{123}, \phi \equiv \phi_{123}$),
    \begin{align*}
        &\frac{\partial }{\partial \Delta_{ij}}\Bigr(|\Delta|^2-\lambda \Bigr(1-\sum_{\{i,j\} \in C_3}\Delta_{ij}+2|\Delta|\cos(\phi)\Bigr)\Bigr) = 0
    \end{align*}
    and therefore, 
    \begin{equation}
        |\Delta_{123}|^2+\lambda(\Delta_{ij}-|\Delta_{123}|\cos(\phi_{123})) = 0.
    \end{equation}

    We conclude that the optimal values are achieved by third-order invariants where all two-state overlaps have the same value $\Delta_{12} = \Delta_{13} = \Delta_{23}$. This implies that Eq.~\eqref{eq: constraint 2} becomes, 
    \begin{equation}
        1-3|\Delta_{123}|^{\frac{2}{3}} + 2|\Delta_{123}|\cos(\phi_{123}) \geq 0
    \end{equation}
    with the boundary given by equality, as we wanted to prove.
\end{proof}

{We can draw two observations from the previous proof. First, the condition that defines the border of the set of physically realizable 3rd order Bargmann invariants is $\det(G_{\Psi})=0$. This implies that the corresponding Gram matrix is not full rank and so there exists a physical physical realization of these matrices using only qubit states. Second, it can be seen that the Gram matrices of dimension 3 that maximize $|\Delta_{123}|$ can be written as circulant matrices, via an appropriate gauge choice. If a Gram matrix $G_{\Psi}$ of the form given in Eq.~\eqref{eq: Bargmann matrix} is such that $\Delta_{12} = \Delta_{13} = \Delta_{23}= |\alpha|$ we can do the gauge transformation $\ket{\psi'_1}= \ket{\psi_1} $, $\ket{\psi'_2}= e^{i \phi_{123}/3}\ket{\psi_2}$, $\ket{\psi'_3}= e^{-i \phi_{123}/3}\ket{\psi_3}$, so that the Gram matrix $G'_{ij}= \langle \psi'_i|\psi'_j \rangle$ is circulant. This realization will help us tackle the problem of characterizing Bargmann invariants of 4th order. By focusing on circulant Gram matrices,  the constraint coming from the fact that Gram matrices are positive semidefinite is easier to analyse, since eigenvalues of circulant matrices have a simple analytical form in terms of the matrix entries. }

We say that a generic $n \times n$ matrix is circulant if every row is the cyclic permutation of the previous row. Let us define the set of all $n \times n$ circulant matrices as $\mathsf{circ}(n)$, and the set of all $n \times n$ circulant matrices that are also correlation Gram matrices as $\mathsf{gramcirc}(n)$. Let us also define,
\begin{equation}    \mathcal{B}_n|_{\mathsf{circ}} := \{\Delta \in \mathcal{B}_n \,|\, \Delta = \Delta_\Psi \Leftrightarrow G_{\Psi} \in \mathsf{gramcirc}(n)\}.
\end{equation}

We have that $\Delta \in \mathcal{B}_4 |_{\mathsf{circ}}$ if, and only if, the Bargmann invariant representation of circulant Gram matrices $H(\alpha,\delta)$ is quantum realizable. We write  $H(\alpha,\delta)$ as
\begin{equation}
    H(\alpha,\delta)=\begin{bmatrix}1 & \alpha & \sqrt{\delta} & \alpha^{*}\\
\alpha^{*} & 1 & \alpha & \sqrt{\delta}\\
\sqrt{\delta} & \alpha^{*} & 1 & \alpha\\
\alpha & \sqrt{\delta} & \alpha^{*} & 1
\end{bmatrix}
\end{equation}
where $\Delta = \alpha^4 \in \mathcal{B}_4 |_{\mathsf{circ}}$ and $\delta \in \mathcal{B}_2$. Therefore, any quantum realization $\underline\psi \in \mathcal{P}(\mathcal{H})^4$ for this matrix must satisfy that $\langle \psi_1|\psi_2 \rangle = \langle \psi_2|\psi_3 \rangle = \langle \psi_3|\psi_4 \rangle = \langle \psi_4|\psi_1 \rangle = \alpha \in \mathbb{C}$ and $\langle \psi_1|\psi_3 \rangle = \langle \psi_2|\psi_4 \rangle = \sqrt{\delta} \in \mathbb{R}$.  

$\Delta$ is quantum realizable if, and only if (from theorem~\ref{theorem: Gram matrix realization}) $H(\Delta,\delta)$  is positive-semidefinite. We use this to bound the possible values of invariants satisfying this condition.

\begin{proof}[Proof of theorem~\ref{theorem: circulant gota}]
    For circulant Gram matrices, the forth-order invariant $\Delta_{1234}$ can be described as

\[
\Delta_{1234}=\alpha^{4}.
\]
Since $H(\alpha,\delta)$ is circulant, its eigenvalues $\lambda$ are given by

\[
\lambda=1+\sqrt{\delta}\pm2\mathrm{Re}\left(\alpha\right),\,1-\sqrt{\delta}\pm2\mathrm{Im}\left(\alpha\right)
\]
Positive semidefiniteness demands that all eigenvalues be non-negative, which is equivalent to demanding

\[
\begin{array}{cc}
2\left|\mathrm{Re}\left(\alpha\right)\right| & \leq 1+\sqrt{\delta}\\
2\left|\mathrm{Im}\left(\alpha\right)\right| & \leq 1-\sqrt{\delta}
\end{array}\Leftrightarrow\begin{array}{cc}
|\alpha| & \leq\frac{1+\sqrt{\delta}}{2\left|\cos\theta\right|}\\
|\alpha| & \leq\frac{1-\sqrt{\delta}}{2\left|\sin\theta\right|},
\end{array}
\]
where $\alpha = |\alpha|e^{i\theta}$. Maximizing $|\alpha|$ subject to the above restrictions gives

\[
\left|\alpha \right|=\min\left(\frac{1+\sqrt{\delta}}{2\left|\cos\theta\right|},\frac{1-\sqrt{\delta}}{2\left|\sin\theta\right|}\right).
\]
This minimization implies,

\[
\left|\alpha \right|=\frac{1+\sqrt{\delta}}{2\left|\cos\theta\right|}=\frac{1-\sqrt{\delta}}{2\left|\sin\theta\right|}\Rightarrow \sqrt{\delta}=\frac{\left|\cos\theta\right|-\left|\sin\theta\right|}{\left|\sin\theta\right|+\left|\cos\theta\right|},
\]
and hence 

\[
\alpha=\frac{e^{i\theta}}{\left|\sin\theta\right|+\left|\cos\theta\right|}.
\]

Because $\alpha^4 =(|\alpha|e^{i\theta})^4=\Delta_{1234} = |\Delta|e^{i\phi}$, substituting $4\theta=\phi$,
we get the boundary of $\mathcal{B}_4|_{\mathsf{circ}}$ is described by the complex curve

\[
\Delta=\frac{e^{i\phi}}{\left(\sin\frac{\phi}{4}+\cos\frac{\phi}{4}\right)^{4}}
\]
with $\phi \in [0,2\pi)$. 
\end{proof}

To conclude, we also show that the boundary of the set $\mathcal{B}_{4}\vert_{\mathsf{circ}}$ is attainable by tuples of vectors in $\mathbb{C}^2$. We can let four-tuples of single-qubit states to have the following form:

\begin{equation}
\left|\psi_{k}\right\rangle =\sin\theta\left|0\right\rangle +\omega^{k}\cos\theta\left|1\right\rangle, 
\end{equation}
where $k=0,1,2,3$. For this tuple $\Psi = (\vert \psi_k \rangle \langle \psi_k \vert )_k$, the Gram matrix $G_\Psi$ is always circulant, and therefore the fourth-order Bargmann invariant $\Delta_4 = \langle \psi_0\vert \psi_1\rangle \dots \langle \psi_3\vert \psi_0\rangle \in \mathcal{B}_4|_{\mathsf{circ}}$.  

This Gram matrix $G_\Psi$ will have elements,

\begin{equation}
a=\left\langle \psi_{k}\left.\right|\psi_{k+1}\right\rangle =\sin^{2}\theta+\omega\cos^{2}\theta.
\end{equation}
where $\omega$ is an $4$-th root of unity. The only option for this root that can give a non-zero imaginary part for the Bargmann invariant is $\omega=i$, or equivalently $\omega=-i$ which gives a complex conjugate Bargmann invariant.

The invariant $\Delta_4=r e^{i\phi}$ is then of the form
\begin{equation}
\Delta_4=\left(\sin^{2}\theta+i\cos^{2}\theta\right)^{4},\label{eq:Invariant}
\end{equation}
with its radius $r = |\Delta|$ given by
\begin{equation}
r=\left(\sin^{4}\theta+\cos^{4}\theta\right)^{2}
\end{equation}
and its phase $\phi$ satisfying
\begin{equation}
\tan\frac{\phi}{4}=\frac{\cos^{2}\theta}{\sin^{2}\theta}=\frac{1}{\sin^{2}\theta}-1.
\end{equation}
Putting the last two equations together we have 

\begin{align}
r & =\left(\frac{1}{\left(1+\tan\frac{\phi}{4}\right)^{2}}+\left(\frac{\tan\frac{\phi}{4}}{1+\tan\frac{\phi}{4}}\right)^{2}\right)^{2}\\
&=\left(\frac{1+\tan^{2}\frac{\phi}{4}}{\left(1+\tan\frac{\phi}{4}\right)^{2}}\right)^{2}\nonumber \\
 & =\left(\frac{1}{\left(1+\tan\frac{\phi}{4}\right)^{2}\cos^{2}\frac{\phi}{4}}\right)^{2}\\
 &=\frac{1}{\left(\cos\frac{\phi}{4}+\sin\frac{\phi}{4}\right)^{4}}.\label{eq:drop equation}
\end{align}

The above equation reproduces exactly to the boundary of the set $\mathcal{B}_4|_{\mathsf{circ}}$, as we wanted to show.

\section{Proof of Theorems~\ref{theorem: 3 overlaps cannot witness} and~\ref{theorem: overlaps can witness}}

We start showing that, for every overlap triple $\pmb\Delta \in \mathcal{Q}_3$ realizable by some tuple $\Psi \in \mathcal{P}(\mathcal{H})^3$ with respect to some space $\mathcal{H}$, there exists another overlap triple ${\Psi^{\mathbb{R}}} \in \mathcal{R}(\mathcal{H}',\mathbb{A})$ for some $\mathcal{H}'$ and some basis $\mathbb{A}$ such that
\begin{equation*}
    \pmb\Delta_\Psi  = \left(\begin{matrix}
        |\langle \psi_1|\psi_2 \rangle |^2\\ |\langle \psi_1|\psi_3 \rangle |^2\\ |\langle \psi_2|\psi_3\rangle |^2
    \end{matrix}\right) = \left(\begin{matrix}
        |\langle \psi_1^{\mathbb{R}}|\psi_2^{\mathbb{R}}|^2\\
        |\langle \psi_1^{\mathbb{R}}|\psi_3^{\mathbb{R}} |^2\\ |\langle \psi_2^{\mathbb{R}}|\psi_3^{\mathbb{R}} |^2
    \end{matrix}\right)= \pmb\Delta_{\Psi^{\mathbb{R}}},
\end{equation*}
meaning that $\mathcal{R}_3 = \mathcal{Q}_3$.

\begin{proof}[Proof of theorem~\ref{theorem: 3 overlaps cannot witness}]
    To start, we assume that all overlaps are non-zero. In this case, every such tuple is quantum realizable iff the associated matrix $H(\pmb\Delta)$ from Eq.~\eqref{eq: Bargmann matrix} is positive-semidefinite, as we have seen in the proof of theorem~\ref{theorem: gota}. For this to be satisfied it is necessary and sufficient that Eqs.~\eqref{eq: constraint 1} and~\eqref{eq: constraint 2} be satisfied. Because Eq.~\eqref{eq: constraint 1} is trivially satisfied by any quantum realizable overlap, we are left with Eq.~\eqref{eq: constraint 2}. If we write this as,
    \begin{align}
        \det(H(\pmb\Delta,\phi)) =  1-\Delta_{12}-\Delta_{13}-\Delta_{23}\nonumber\\+2\sqrt{\Delta_{12}\Delta_{13}\Delta_{23}}\cos(\phi) \geq 0 \label{eq: triples overlap proof determinant}
    \end{align}
    this inequality corresponds to the bounds that must be satisfied by any quantum realization $\pmb\Delta \in \mathcal{Q}_3$. Note that imaginarity of any triple is completely captured by the phase of the third-order invariant $\Delta_{123}$, which is the imaginarity involved in this bound. If we now assume that $\phi \in \{0,\pi\}$ we get the other two bounds,
    \begin{equation}\label{eq: zero bound triple case}
        \det(H(\pmb\Delta,0)) = 1-\Delta_{12}-\Delta_{13}-\Delta_{23}+2\sqrt{\Delta_{12}\Delta_{13}\Delta_{23}} \geq 0,
    \end{equation}
    \begin{equation}\label{eq: pi bound triple case}
        \det(H(\pmb\Delta,0)) = 1-\Delta_{12}-\Delta_{13}-\Delta_{23}-2\sqrt{\Delta_{12}\Delta_{13}\Delta_{23}} \geq 0
    \end{equation}
    We now note that Ineq.~\eqref{eq: zero bound triple case} is satisfied by all possible triples of overlaps, as was shown in Ref.~\cite{galvaobrod2020quantum}. In particular, it is clear that if Eq.~\eqref{eq: triples overlap proof determinant} is satisfied Eq.~\eqref{eq: zero bound triple case} must also be satisfied. But this implies that if $H(\pmb\Delta,\phi)$ is quantum realizable then $H(\pmb\Delta,0)$ is \emph{also} quantum realizable. However, any quantum realization of $H(\pmb\Delta,0)$ is \emph{real} with respect to some Hilbert space and some basis. In particular, the Cholesky decomposition of $H(\pmb\Delta,0)$ will be one such real representation.

    To conclude the proof, we analyse the cases where some overlap is equal to zero. Suppose that $\Delta_{12} = 0$ without loss of generality. This implies that for any quantum realization $\Psi$, it is always possible to take $\vert \psi_1\rangle = \vert 0\rangle, \vert \psi_2 \rangle = \vert 1\rangle$. Because of that, both $\Delta_{13}$ and $\Delta_{23}$ will have some imaginarity-free quantum realization.
\end{proof}

Surprisingly, this is a fact that is true for triples of quantum states, but not for tuples of 4 states or more.

\begin{proof}[Proof of theorem~\ref{theorem: overlaps can witness}]
    We give a counter-example. Let us consider the matrix 
    \begin{widetext}
    \begin{equation}
        H(\pmb\Delta,\phi_{123},\phi_{124},\phi_{134}) = \left(\begin{matrix}
            1 & \sqrt{\Delta_{12}} & \sqrt{\Delta_{13}} & \sqrt{\Delta_{14}}\\
            \sqrt{\Delta_{12}} & 1 & \sqrt{\Delta_{23}}e^{i\phi_{123}} & \sqrt{\Delta_{24}}e^{i\phi_{124}}\\
            \sqrt{\Delta_{13}} & \sqrt{\Delta_{23}}e^{-i\phi_{123}} & 1 & \sqrt{\Delta_{34}}e^{i\phi_{134}}\\
            \sqrt{\Delta_{14}} & \sqrt{\Delta_{24}}e^{-i\phi_{124}} & \sqrt{\Delta_{34}}e^{-i\phi_{134}} & 1
        \end{matrix}\right).
    \end{equation}
    \end{widetext}
    There will be some quantum realization for this matrix, if and only if, it is positive-semidefinite. Because the imaginarity is completely determined by the values $\phi_{123},\phi_{124},\phi_{134}$ we let these to be only elements of the set $\{0,\pi\}$. As in the proof of theorem~\ref{theorem: 3 overlaps cannot witness}, if for any quantum realizable Gram matrix, at least one of all possible 8 real-only matrices $\{H(\pmb\Delta,\kappa_{123},\kappa_{124},\kappa_{134})\}_{\kappa_{1ij} \in \{0,\pi\}}$ is quantum realizable, we would have a similar result as the one from theorem~\ref{theorem: 3 overlaps cannot witness}. A counter example (many exist!) is the tuple $\underline\psi \in \mathcal{P}(\mathbb{C}^2)^4$:
    \begin{align}
        \vert \psi_1\rangle &= \vert 0\rangle \\
        \vert \psi_2 \rangle &= \vert +\rangle \\
        \vert \psi_3\rangle &= \vert -i \rangle \\
        \vert \psi_4\rangle &= \cos(\pi/6)\vert 0\rangle + e^{i\pi/4}\sin(\pi/6)\vert 1\rangle 
    \end{align}
    We have in this case, 
    \begin{align*}
    &\Delta_{12}(\underline\psi) = \frac{1}{2}, \Delta_{13}(\underline\psi) = \frac{1}{2}, \Delta_{14}(\underline\psi) = \frac{3}{4}\\
    &\Delta_{23}(\underline\psi) = \frac{1}{2}, \Delta_{24}(\underline\psi) = \frac{4+\sqrt{6}}{8}, \Delta_{34}(\underline\psi) = \frac{4-\sqrt{6}}{8}
\end{align*}
If we now look at the eigenvalues of all the matrices $\{H(\pmb\Delta,\kappa_{123},\kappa_{124},\kappa_{134})\}_{\kappa_{1ij} \in \{0,\pi\}}$ we see that the smallest eigenvalue is always \emph{negative},
\begin{align}
    \lambda_{min}(G(0,0,0)) &= -0.044984,\label{eq: eigenvalue 1}\\
    \lambda_{min}(G(\pi,0,0)) &= -0.512315,\\
    \lambda_{min}(G(0,\pi,0)) &= -0.709002,\\
    \lambda_{min}(G(0,0,\pi)) &= -0.561292,\\
    \lambda_{min}(G(\pi,\pi,0)) &= -0.837603,\\
    \lambda_{min}(G(0,\pi,\pi)) &= -0.704281,\\
    \lambda_{min}(G(\pi,0,\pi)) &= -0.491359,\\
    \lambda_{min}(G(\pi,\pi,\pi)) &= -1.17472. \label{eq: eigenvalue 8}
\end{align}
implying that all these matrices are \emph{not} positive-semidefinite. In other words, there can be no quantum realization with real-values only for these six overlaps $\pmb\Delta = (\Delta_{12},\Delta_{13},\Delta_{14},\Delta_{23},\Delta_{24},\Delta_{34})$ because otherwise, one of these matrices would need to be quantum realizable.

To see that this result is independent of gauge choice, leading to the particular form of $H(\pmb\Delta,\phi_{123},\phi_{124},\phi_{134})$ we can look at all possible matrices
\begin{widetext}
    \begin{equation}
        H(\pmb\Delta,\pmb\kappa) = \left(\begin{matrix}
            1 & \sqrt{\Delta_{12}}e^{i\kappa_1} & \sqrt{\Delta_{13}}e^{i\kappa_2} & \sqrt{\Delta_{14}}e^{i\kappa_3}\\
            \sqrt{\Delta_{12}}e^{-i\kappa_1} & 1 & \sqrt{\Delta_{23}}e^{i\kappa_4} & \sqrt{\Delta_{24}}e^{i\kappa_5}\\
            \sqrt{\Delta_{13}}e^{-i\kappa_2} & \sqrt{\Delta_{23}}e^{-i\kappa_4} & 1 & \sqrt{\Delta_{34}}e^{i\kappa_6}\\
            \sqrt{\Delta_{14}}e^{-i\kappa_3} & \sqrt{\Delta_{24}}e^{-i\kappa_5} & \sqrt{\Delta_{34}}e^{-i\kappa_6} & 1
        \end{matrix}\right).
    \end{equation}
    \end{widetext}
If one finds all eigenvalues for all $2^6$ possible choices of $\pmb\kappa$ we see that all matrices return the smallest eigenvalues that are, still, those given by Eq.~\eqref{eq: eigenvalue 1}-\eqref{eq: eigenvalue 8}. Therefore, all $2^6$ matrices have negative eigenvalues, and independent of gauge choice, $H(\pmb\Delta,\pmb\kappa)$ cannot have a quantum realization.
\end{proof}

We have numerically investigated how often it is possible to witness imaginarity using the procedure above. We sample $10^6$ Haar random tuples of four states of different dimensions, and check the fraction of tuples that are successfully witnessed. For two-dimensional quantum state tuples, $608\, 329$ tuples were witnessed using the methodology above. For three-dimensional quantum states, $59\,803$ and for four-dimensional quantum states only $5 \, 914$. This suggests that as the dimension increases, a smaller fraction of uniformly drawn Haar random tuples have their imaginarity signalled by the witnesses just discussed.

\end{document}

%% file: preamble.tex
\usepackage{amsthm}
\usepackage{mathtools}
\usepackage{physics}
\usepackage{bbm}
\usepackage{amsfonts}
\usepackage{amssymb}
\usepackage{graphicx}

\usepackage[T1]{fontenc}
\usepackage{bbold}
\usepackage{float}
\usepackage[utf8]{inputenc}
\usepackage{csquotes}
\usepackage{tikz-cd}
\usepackage[normalem]{ulem}
\usetikzlibrary{shapes, shapes.geometric, shapes.symbols, shapes.arrows, shapes.multipart, shapes.callouts, shapes.misc,decorations.pathmorphing}
\tikzset{snake it/.style={decorate, decoration=snake}}
\usepackage{tcolorbox}

\usepackage[colorlinks=true,linktocpage=true]{hyperref}
\hypersetup{
    allcolors  = {blue},
}

\theoremstyle{plain}
\newtheorem{theorem}{Theorem}

\definecolor{selectiveyellow}{rgb}{0.8, 0.0, 0.8}

\usepackage{array}
\usepackage{multirow}
\usepackage{xcolor}

\usepackage{tikz}